\documentclass[a4paper]{article}

\usepackage{a4wide}

\usepackage{amsmath,amssymb,amsfonts,amsthm,mathrsfs}

\usepackage{float}
\usepackage{caption, subcaption}

\usepackage{tikz,pgfplots}
\pgfplotsset{compat=1.4}
\usetikzlibrary{shapes,calc,arrows,positioning,decorations.markings}

\usepackage{enumerate}

\numberwithin{equation}{section} 

\newtheorem{theorem}{Theorem}[section]
\newtheorem{lemma}[theorem]{Lemma}
\newtheorem{corollary}[theorem]{Corollary}

\def\N{{\mathbb N}}
\def\Z{{\mathbb Z}}
\def\C{{\mathbb C}}
\def\R{{\mathbb R}}
\def\DD{{\mathbb D}}
\def\bbS{{\mathbb S}}

\def\cH{{\mathcal H}}
\def\cB{{\mathcal B}}
\def\cC{{\mathcal C}}
\def\cM{{\mathcal M}}
\def\cA{{\mathcal A}}
\def\cF{{\mathcal F}}
\def\cE{{\mathcal E}}
\def\cD{{\mathcal D}}

\def\ri{{\mathrm{i}}}
\def\rd{{\mathrm{d}}}
\def\bbI{{\mathrm{Id}}}

\newcommand{\sB}{\mathscr{B}}
\newcommand{\sC}{\mathscr{C}}

\def\bk{{\mathbf k}}
\def\br{{\mathbf r}}
\def\bK{\boldsymbol{\omega}}
\def\bnull{{\mathbf 0}}

\def\Tr{{\rm Tr} \, }
\def\Ran{{\rm Ran} \, } 

\def\dist{{\rm dist}}
\def\SU{{\rm SU}} 
\def\U{{\rm U}} 
\def\obs{{\rm obs}}
\def\TRIM{{\rm TRIM}}

\newcommand{\BZ}{{\mathcal B}}
\newcommand{\eps}{\varepsilon}
\newcommand{\ie}{{\sl i.\,e.\ }}   
\newcommand{\eg}{{\sl e.\,g.\ }} 
\DeclareMathOperator*{\Ch}{Ch}
\DeclareMathOperator*{\Span}{Span}

\title{Localised Wannier functions in metallic systems}
\author{Horia Cornean, David Gontier, Antoine Levitt, Domenico Monaco} 

\usepackage[colorlinks, linkcolor={blue!50!blue}, citecolor={red}]{hyperref}

\begin{document}  
\maketitle  
\begin{abstract}
The existence {and construction} of exponentially localised Wannier functions for insulators is a well-studied problem. {In comparison}, the case of metallic systems has been much less explored, even though localised Wannier functions constitute an important and widely used tool for the numerical band interpolation of metallic condensed matter systems. In this paper we prove that, under generic conditions, $N$ energy bands of a metal can be exactly represented by $N+1$ Wannier functions decaying faster than any polynomial. We also show {that}, in general, the lack of a spectral gap does not allow for exponential decay.


\end{abstract}

\section{Introduction}

In this paper we consider the problem of constructing Wannier functions for metallic systems. We start from a system of independent electrons in a periodic crystal in physical dimension $d \le 3$, whose properties are characterised by the Schr{\"o}dinger operator $H = -\Delta + V$, where $V$ is a periodic potential. By Bloch theory, the spectral properties of $H$ can be studied through its Bloch fibers $H(\bk) = (- \ri \nabla + \bk)^{2} + V$ with periodic boundary conditions on a unit cell of the crystal. In particular, the eigenvalues $\varepsilon_{n\bk}$ of $H(\bk)$ give access to the band structure of the crystal. As it is computationally expensive to diagonalise $H(\bk)$ for a given $\bk$, it is numerically desirable to interpolate the eigenvalues $\varepsilon_{n\bk}$ efficiently, so that a limited sampling of the Brillouin zone can still yield an accurate representation of the band structure.

For a fixed $n \in \N$, $\varepsilon_{n\bk}$ is periodic as a function of $\bk$. Therefore a very natural interpolation method is Fourier interpolation, where the values of $\varepsilon_{n\bk}$ on a finite equispaced grid are used to determine the first Fourier coefficients of the function, which can then be used to determine $\varepsilon_{n\bk}$ at any other point outside the grid. This method is very efficient if the functions $\varepsilon_{n\bk}$ are smooth, which corresponds to a quick decay of their Fourier coefficients. A major difficulty in this approach is induced by eigenvalue crossings; at such points, the eigenvalues are no longer differentiable. While in dimension $1$ this can locally be fixed by a relabelling of the eigenvalues, this is not possible in higher dimensions. These crossings produce Gibbs-like oscillations in the Fourier interpolation and limit its practical efficiency.

This can be remedied by noting that while the eigenvalues themselves have discontinuous first derivatives, they can be expressed as eigenvalues of a smooth matrix of much smaller size. For instance, in the insulating case, the first $N$ bands are isolated from the rest. Assume that we can find a smooth basis  (not necessarily consisting of eigenvectors) $u_{1\bk},\dots,u_{N\bk}$, which satisfies the appropriate pseudo-periodicity conditions and spans the subspace associated to the first $N$ eigenvalues of $H(\bk)$. Then we can build the $N\times N$ Hermitian matrix $A_{mn}(\bk) = \langle u_{m\bk}, H(\bk) u_{n\bk}\rangle$, which is smooth and periodic. Given a point $\bk$ not on the grid, we can use Fourier interpolation to interpolate the small-dimension matrix $A$ at $\bk$, and diagonalise it to recover an accurate approximation to $\varepsilon_{1\bk},\dots,\varepsilon_{N\bk}$. Equivalently, this procedure can be seen as building a tight-binding model, whose parameters, the Fourier coefficients of $A_{mn}(\bk)$, can be written as matrix elements of $H$ in the basis of Wannier functions \cite{yates2007spectral}. The smoothness of $u_{n\bk}$ (or equivalently the localisation properties of the Wannier functions) determine the effectiveness of the interpolation scheme. It is therefore important in applications to establish the existence of localised Wannier functions and design algorithms to construct them.

For insulating systems, constructing localised Wannier functions amounts to finding a smooth and pseudo-periodic basis of the spectral subspace associated with the first $N$ bands. This is not always possible because of the competition between smoothness and pseudo-periodicity. Accordingly, the possibility of constructing localised Wannier functions is characterised by a set of topological invariants, the Chern numbers \cite{brouder2007exponential,panati2007triviality}. When these vanish, as is the case for instance for systems with time-reversal symmetry (TRS), it is possible to construct exponentially localised Wannier functions. Different algorithms exist to construct them, either by optimising a smoothness criterion starting from a physical initial guess \cite{marzari1997maximally,mostofi2008wannier90,mustafa2015automated}, or by building a smooth gauge from the ground up \cite{fiorenza2016construction,cornean2016construction,cornean2017_3D,cances2017robust,Gontier2017todo,damle2017scdm}.

We consider in this paper the case of \emph{metallic} systems, where the bands are not isolated. The construction of Wannier functions for non-isolated bands has been proposed in \cite{souza2001maximally}, and successfully used over many years to compute properties of materials (see \eg \cite{marzari2012maximally} for a review). In this case, one is interested in reproducing the band structure of the first $N$ bands. In general, they cross with higher bands, and it is not possible to find $N$ localised Wannier functions representing them. One then tries to find $N+K$ Wannier functions that represent the first $N$ bands, in the sense that they span a superspace of the spectral subspace associated to the first $N$ bands. Numerically, this yields a $(N+K) \times (N+K)$ matrix $A(\bk)$, whose $N$ lowest eigenvalues are $\varepsilon_{1\bk},\dots,\varepsilon_{N\bk}$, and whose other $K$ eigenvalues are ignored. As in the insulating case, this small matrix can be interpolated efficiently to recover the whole band structure.

In contrast with the spectrally isolated case, nothing is known theoretically about this procedure, and it is not \textit{a priori} clear what the localisation properties of these Wannier functions are. In this paper, we prove that, under conditions satisfied by most systems (see Assumptions 1 and 2 in Section~\ref{sec:main}), one can always find $N+1$ localised Wannier functions which decay faster than any polynomial (almost exponential localisation), and which represent the first $N$ bands in the above sense. We prove that this is optimal, in the sense that exponential localisation is impossible in general. Note that our results prove the existence of almost exponentially localised disentangled Wannier functions in the sense of \cite{souza2001maximally}. Importantly, it does not prove that the method of minimising the spread of \cite{souza2001maximally} produces almost-exponentially localised Wannier functions: indeed, this fact was recently found to be false in general \cite{damle2017}. We refer to Section \ref{sec:discuss} for a discussion.

The paper is structured as follows. In Section~\ref{sec:main} we present our generic assumptions and state our results. We provide a brief but self-contained introduction to Chern numbers and their relevance to the existence of  smooth extensions of families of projections in Section~\ref{sec:chern}.  The proofs of our main results are gathered in Section~\ref{sec:proof}. Finally,   we show in Appendix~\ref{app:analyticity} that the construction of exponentially localised Wannier functions is impossible in general.

\paragraph{Acknowledgements}

Financial support from Grant 4181-00042 of the Danish Council for Independent Research $|$ Natural Sciences, from the ERC Consolidator
Grant 2016 ``UniCoSM - Universality in Condensed Matter and Statistical Mechanics'', and from PEPS JC 2017 is gratefully acknowledged.

\section{Main results} \label{sec:main}

\subsection{Statement of the problem}
Let $\cH$ be a complex Hilbert space of dimension $M \in \N \cup \{ + \infty \}$. We consider a family of self-adjoint operators $H(\bk)$ acting on $\cH$ parametrised by $\bk \in \R^d$ with $d \le 3$, and satisfying the following properties:
\begin{enumerate}[(i)]
\item \label{item:smooth} the map $\bk \to H(\bk)$ is smooth%
\footnote{In the case $M = \infty$ and $H(\bk)$ unbounded, we assume that there exists $\lambda \in \R$ such that $H(\bk)$ is bounded from below by $\lambda$ for all $\bk \in \R^{d}$ and that $\bk \mapsto (H(\bk)-\lambda)^{-1}$ is compact and smooth.}%
;
\item \label{item:periodic} {the map $\bk \to H(\bk)$ is $\Z^d$-periodic}.
\end{enumerate}

Here and in the rest of the paper, smooth means infinitely differentiable. These conditions naturally appear in the study of Schr{\"o}dinger operators $H = -\Delta + V$ with $V$ being a $2\pi\Z^{d}$-periodic potential, see \eg \cite[Theorems XIII.97 and XIII.99]{reed1978analysis}, where it is proved that $H$ is unitarily equivalent to $\int_{[-1/2,1/2]^d}^\oplus H(\bk) \rd\bk$, where the unbounded ($\Z^d$-periodic) fiber operator $H(\bk)=-\Delta+V$ acts on the Hilbert space $L^2([-\pi,\pi]^d)$ with $\bk$-twisted boundary conditions. This operator is often studied with a change of variables: after a unitary rotation $\tau_\bk$ given by the multiplication with $e^{-\ri \bk\cdot \br}$ with $\br\in [-\pi,\pi]^d$, one generates a new fiber Hamiltonian $\widetilde{H}(\bk):=\tau_\bk H(\bk)\tau_\bk^{-1} = (-\ri \nabla + \bk)^{2} + V$ having periodic boundary conditions but which is no longer periodic in $\bk$. In this case, periodicity (\ref{item:periodic}) is replaced by the covariance property
\begin{enumerate}
\item[(ii')] \label{item:periodic2} for all $\mathbf{K} \in \Z^d$, $\widetilde{H}(\bk + \mathbf{K}) = \tau_{\mathbf{K}} \widetilde{H}(\bk) \tau_{\mathbf{K}}^{-1}$.
\end{enumerate}
It is possible to deal with this covariance property directly, but to simplify the geometric arguments in this paper we will only consider the $\bk$-periodic picture, \ie $H(\bk)$. In order to avoid working with unbounded operators and in particular with the fact that the domain of $H(\bk)$ is $\bk$-dependent, we can instead study the smooth, compact and periodic family defined by the resolvent $(H(\bk) - \lambda)^{-1}$, acting on $L^2([-\pi,\pi]^d)$. The conditions above on $H(\bk)$ with finite $M$ also cover the case of discrete Schr\"odinger operators and tight-binding models of common use in computational solid state physics.

We sometimes consider the extra condition that $H$ is time-reversal symmetric (TRS) {with respect to a time-reversal operator of bosonic type}, which means that there exists an antiunitary operator $\theta \colon \mathcal{H} \to \mathcal{H}$ such that {$\theta^2= {\rm Id}$} and 
\begin{enumerate}
	\item[(iii)] for all $\bk \in \R^{d}$, $H(-\bk) = \theta H(\bk) \theta^{-1}$.
\end{enumerate}
The TRS property holds for a system without magnetic fields and when spin is ignored, in which case $\theta$ acts simply as complex conjugation on the wavefunctions.

We work in the Brillouin zone $\BZ := \R^{d} / \Z^{d}$, equipped with the topology of a torus. In particular, properties~\eqref{item:smooth} and~\eqref{item:periodic} are equivalent to saying that $H$ defines a smooth map on the torus $\BZ$. We denote by $\varepsilon_{1}(\bk) \le \varepsilon_{2}(\bk) \le \cdots \leq \eps_{M}(\bk)$ the eigenvalues of $H(\bk)$ labelled in increasing order, counting multiplicities, and by $K_n \subset \BZ$ the set of crossings between the $n$-th and $(n+1)$-th bands:
\[ K_n :=  \left\{\bk \in \BZ, \ \varepsilon_{n\bk} = \varepsilon_{n+1,\bk} \right\}. \]
Outside of the crossing set $K_n$, we can define the projection $P_n(\bk)$ on the eigenspace corresponding to the $n$ lowest eigenvalues. If $u_{m \bk}$ denotes the $m$-th Bloch function, so that $H(\bk) u_{m\bk} = \varepsilon_{m\bk} u_{m\bk}$, then
\begin{equation} \label{eq:def:Pn} 
\forall \bk \in \cB \setminus K_n, \quad P_{n}(\bk) = \sum_{m \le n} \left| u_{m\bk} \right\rangle \left\langle u_{m\bk} \right| = \dfrac{1}{2 \pi \ri} \oint_{\sC} \left (z - H(\bk)\right )^{-1} {\rd z},
\end{equation}
where $\sC \equiv \sC_{\bk}$ is a contour {in the complex energy plane} that encloses $\{ \varepsilon_{1}(\bk), \cdots, \varepsilon_{n}(\bk) \}$ but not $\{\varepsilon_{n+1}(\bk),\dots\}$. The contour can be chosen to be locally constant in $\bk$. Then, this last formula shows that the map $\bk \mapsto P_{n}(\bk)$ is smooth on $\BZ \setminus K_n$, even though the individual Bloch functions are not smooth at crossings between the bands $\varepsilon_1,\dots,\varepsilon_n$. Our goal is the construction of a \emph{Bloch frame}, \ie a smooth orthonormal basis of $\Ran P_{n}(\bk)$. In the context of periodic Schr\"{o}dinger operators, the inverse Bloch transform of such a frame yields \emph{localised Wannier functions}.

When studying gapped systems, or {\em insulators}, with $N$ electrons per unit cell and Fermi level $\varepsilon_{F}$, we have $\sup_{\bk \in \BZ} \varepsilon_{N,\bk} < \varepsilon_{F} < \inf_{\bk \in \BZ} \varepsilon_{N+1,\bk}$, and so in particular $K_N = \emptyset$. In this case, $P_{N}(\bk)$ is well-defined and smooth on the whole Brillouin zone $\BZ$. The existence of a smooth Bloch frame for $P_N$ is then equivalent, {in two and three dimensions,} to the vanishing of the Chern class of the Bloch bundle associated with $P_{N}$, a condition that is always satisfied whenever $H$ (and hence $P_N$) is TRS~\cite{panati2007triviality,brouder2007exponential,monaco2015symmetry,monaco2017chern}.

In the present article we are interested in the case of {\em metallic  systems}, in which case $K_{N}$ is not empty. In this case, the family of projections $P_N$ is not smooth over the whole Brillouin zone $\BZ$. It is therefore impossible to represent $\Ran P_{N}$ with a smooth Bloch frame of rank $N$. As was already discussed in the introduction, an exact representation of that space is often not needed for numerical purposes; rather, as is the case with Wannier interpolation \cite{yates2007spectral}, one only needs a smooth representation of a space containing $\Ran P_{N}$. For instance, if $K_{N+1} = \emptyset$, then we could simply use $P_{N+1}$. However, if bands $N+1$ and $N+2$ cross ($K_{N+1} \neq \emptyset$), $P_{N+1}$ will not be smooth; in general, all the bands might cross and there might not exist a $K > 0$ such that $P_{N+K}$ is smooth.

Following \cite{souza2001maximally}, we consider {\em disentangled Wannier functions}, and look for a projector $P$ of rank $N+K$ whose range contains the one of $P_N$, and which is smooth over the whole Brillouin zone $\BZ$. In the presence of high-degeneracy crossings, $K$ might have to be large to ensure smoothness of $P$. However, by increasing $N$, one can always assume that $K_{N} \cap K_{N+1} = \emptyset$, in which case $K=1$ is enough. Therefore, in the following we will look for $P$ of rank $N+1$. 

\subsection{Assumptions and statement of the main results}
We work in dimension $d = 3$, and exclude severely degenerate cases (see the discussion below in Section~\ref{sec:discuss}). More  precisely, we make the following assumption:
\begin{center}
{\bf Assumption 1:} {\em The sets $K_N$ and $K_{N+1}$ are finite unions of isolated points and piecewise smooth curves, and $K_{N} \cap K_{N+1} = \emptyset$.}
\end{center}

We will also be interested in limiting the span of $P$. This is common practice numerically, where only a finite number of bands are computed, and the range of $P$ is constrained to lie in a prescribed set of bands (called an ``outer window''). We will impose that $\Ran P(\bk) \subset \Ran P_{N+2}(\bk)$, and hence require the following assumption:
\begin{center}
{\bf Assumption 2:} {\em The set $K_{N+2}$ is a finite union of isolated points and piecewise smooth curves, and $K_{N+1} \cap K_{N+2} = \emptyset$.}
\end{center}

We are finally in position to state our main result. Recall that $d = 3$.

\begin{theorem} \label{th:main} 
Under Assumption 1, there exists a family of rank-$(N+1)$ projectors ${P} = P(\bk)$ which is smooth on $\BZ$ and such that $\Ran P_{N}(\bk) \subset \Ran {P}(\bk)$ for all $\bk \in \BZ \setminus K_{N}$.

In addition, if Assumption 2 also holds, then $P$ can be chosen such that $\Ran P(\bk) \subset \Ran P_{N+2} (\bk)$ for all $\bk \in \BZ \setminus K_{N+2}$. 

Finally, if the system is TRS, then ${P}$ can be chosen TRS.
\end{theorem} 

In the TRS case, and in view of the well-developed theory of insulators \cite{panati2007triviality,brouder2007exponential}, this implies the existence of a Bloch frame for $P$.
\begin{corollary}
Under Assumption 1, and if the system is TRS, then there exists a rank-$(N+1)$ TRS {generalised Bloch frame}, \ie a set of $(N+1)$ orthonormal functions ${\Phi}(\bk) = ({\phi_1}, \ldots, {\phi_{N+1}})(\bk)$ which are smooth on $\BZ$, satisfy the TRS condition $\Phi(-\bk) = \theta \Phi(\bk)$, and such that $\Ran P_{N} \subset \Span \Phi$ for all $\bk \in \BZ \setminus K_{N}$. 

In addition, if the system also satisfies Assumption 2, then $\Phi(\bk)$ may be chosen in $\Ran P_{N+2}(\bk)$ for all $\bk \in \BZ \setminus K_{N+2}$.
\end{corollary}

The Bloch frame mentioned in the above statement is a set of $(N+1)$ orthogonal TRS vectors spanning the range of $P_{N}$. In particular, their Bloch transforms are real disentangled Wannier functions. By standard arguments, these Wannier functions decay faster than any polynomial (almost-exponential localisation). This is optimal, in the sense that exponential localisation is generically not possible, as we show in Appendix~\ref{app:analyticity}.

The proof of Theorem~\ref{th:main} is deferred to Section~\ref{sec:proof}, where we first prove the statement without the TRS assumption, and then adapt the strategy to preserve this symmetry. In order to make the argument self-contained, we review some well-known facts about families of projections and their topological invariants in Section~\ref{sec:chern}.

\subsection{Discussion on assumptions and statements} \label{sec:discuss} 
Before going into the proof of our main result, we discuss our assumptions, and the relation to the existing literature on disentangled Wannier functions.

\paragraph{Assumptions}

The choice of focusing on the three-dimensional case is motivated by the von Neumann--Wigner theorem \cite{neumann2000behaviour}, which states that band crossings generically do not appear in one- and two-dimensional systems. However, realistic systems do exhibit such degeneracies due to symmetries (such as the symmetry-protected Dirac cones~\cite{fefferman2012honeycomb}). Although we do not do it explicitly, our Theorem~\ref{th:main} easily generalises to these one- and two-dimensional systems as well, and yields the existence of localised Wannier functions when TRS is present. Moreover, by allowing crossing on curves, we are able to treat most symmetry-imposed crossings along lines of high symmetry in three-dimensional systems as well. As far as we know, the only realistic system not covered by Assumption 1 is the free electron gas.

We formulated Theorem \ref{th:main} assuming that the crossing sets are only composed of isolated points and curves. In fact, the minimal assumption that we make in the proof is that there exists an open subset $\Omega$, with smooth boundary, which contains $K_{N+1}$ and such that $\Omega \cap K_N = \emptyset$. This is certainly satisfied under the condition formulated in Assumption 1, but the proof can be generalised to a slightly bigger class of systems, allowing for more general degeneracies, for example surfaces in $K_N$ which do not separate the points in $K_{N+1}$.

\paragraph{Exponential localisation}

In Theorem \ref{th:main} we only prove the existence of a smooth family of projectors and Bloch frames, yielding Wannier functions decaying faster than any polynomial. To obtain exponential localisation, one would need analyticity. This is much stronger than smoothness because of the rigidity properties of analytic functions, and in Appendix \ref{app:analyticity}, we prove that, if there is a Weyl point in the crossing set $K_{N}$ between bands $N$ and $N+1$, then one cannot find a set of $N+1$ exponentially-localised Wannier functions representing the first $N$ bands. Our proof does not generalise to the case of $N+2$ Wannier functions, and it is an interesting question to determine whether exponential localisation is possible in this case.

\paragraph{Relationship to maximally-localised Wannier functions}

The definition of metallic Wannier functions we use corresponds to the one used in the numerical scheme of \cite{souza2001maximally}, and our results therefore appear to validate this method. It is however important to note that we only prove the \textit{existence} of localised Wannier functions, and not that the Wannier functions found in \cite{souza2001maximally} by minimising the variance (the so-called maximally-localised Wannier functions or MLWFs) are almost-exponentially localized. In fact, it was recently shown that MLWFs for the free electron gas in one and two dimensions correspond to a gauge that is continuous but does not have continuous first derivatives \cite{damle2017}, yielding a decay of the Wannier functions as $|\br|^{-2}$. We expect these results to hold generically for more complicated systems. It is an interesting open question to find a numerical scheme that yields almost-exponentially localised Wannier functions for metallic systems (see \cite{damle2017} for the case of the one-dimensional free electron gas).

Let us also notice that there is a slight difference in our setup compared to that of \cite{souza2001maximally}. In our case, we look for a projector that spans a given number of bands $N$, while in practice one usually tries to reproduce an energy window (the so-called {\em frozen inner window}). Our results can easily be adapted to reproduce a given window by choosing $N$ appropriately. The fact that the projector $P(\bk)$ may be chosen in the range of $P_{N+2}(\bk)$ also gives a natural {\em outer window} of energy.

Finally, we only consider the case where the bands of interest are the $N$ first bands of $H$. Our results extend trivially to the case where the bands of interest lie inside the spectrum of $H$, and are to be disentangled from below as well as from above. In this case, we show that the addition of two extra Wannier functions (one to disentangle the bands from below, and one from above) suffices to disentangle the $N$ bands of interest: namely, we prove the existence of $N+2$ localised Wannier functions that span those bands. In this case, due to the absence of a Rayleigh-Ritz principle, it is not as easy as before to recover the $N$ eigenvalues of interest, as they are not the smallest eigenvalues of the matrix $A_{mn}(\bk) = \langle u_{m\bk},H(\bk) u_{n\bk}\rangle$.

\paragraph{Parseval frames}

Without time-reversal symmetry, the existence of localised Wannier functions for the rank-$(N+1)$ projection $P$ might be topologically obstructed. In this case, if one is willing to give up orthonormality, then it is possible to span the range of $P$ with a ``redundant basis'' (more formally, a \emph{Parseval frame}) consisting of $N+2$ Bloch functions, which are smooth and periodic as functions of $\bk$ \cite{cornean2017parseval,auckly2017parseval}.

\section{Chern numbers and extensions of projectors} \label{sec:chern}

We provide in this section a brief introduction to the concepts of Berry curvature and Chern numbers (see also~\cite{simon1983holonomy}). These objects characterise the possibility of extending a family of projectors defined on the boundary of a three-dimensional open set to its interior. This plays a central role in our proof of Theorem \ref{th:main}. 

Let $\Omega$ be an open set in $\R^3$, and let $\Phi(\bk) = (\phi_1, \cdots, \phi_n)(\bk) \in \cH^n \simeq \cM_{M \times n}(\C)$ be a smooth family on $\Omega$ of $n$ orthonormal vectors (we will write $\Phi^* \Phi= \bbI_n$). The rank-$n$ projector on $\Span \left\{ \phi_1(\bk), \cdots, \phi_n(\bk) \right\}$ is $P := \Phi \Phi^*$, and we say that $\Phi$ is a frame for $P$. The (non-abelian) {\em connection} of the frame $\Phi$ is the matrix-valued $1$-form
\[
\cA[\Phi] := - \ri \Phi^{*} \rd \Phi, \quad \text{so that} \quad \cA_{k,l}[\Phi] (\bk) = \sum_{i=1}^3 - \ri  \langle \phi_k(\bk), \partial_{k_i} \phi_l(\bk) \rangle \rd k_i.
\]
The (abelian) {\em Berry curvature} of the frame $\Phi$ is the real-valued $2$-form $\cF[\Phi] := \rd \Tr \cA[\Phi]$. Since the Berry curvature at a point $\bk$ quantifies the amount of Berry phase produced in an infinitesimally small loop around $\bk$, a quantity which is gauge-invariant, the Berry curvature may be expressed with respect to the gauge-invariant projectors $P(\bk)$ only (without reference to the underlying gauge-dependent frame $\Phi$). Indeed, a straightforward calculation shows that 
\begin{align*}
\cF[\Phi] & = \rd \Tr \cA[\Phi] =  - \ri  \Tr(  (\rd \Phi^*) \wedge (\rd \Phi )  ) = - \ri  \Tr(\Phi \Phi^* \left[ (\rd \Phi ) \Phi^* + \Phi (\rd \Phi^*)\right] \wedge \left[ (\rd \Phi ) \Phi^* + \Phi (\rd \Phi^*)\right] )\\
&  = - \ri  \Tr(P \rd P \wedge \rd P).
\end{align*}

We therefore extend the notion of Berry curvature to any smooth family of projectors $P(\bk)$ by $\cF[P] := - \ri \Tr(P \rd P \wedge \rd P)$. We summarise a few well-known properties of this curvature, which we prove for completeness.

\begin{lemma} \label{lem:sumBerry}
\begin{enumerate}[(i)]
\item (Bianchi's identity) Let $P(\bk)$ be a smooth family of orthogonal projectors. Then the Berry curvature of $P$ is closed (\ie $\rd \cF[P] = 0$).
\item Let $P(\bk)$ and $Q(\bk)$ be smooth families of orthogonal projectors such that $P \perp Q$. Then $ \cF[P+Q] = \cF[P] + \cF[Q]$.
\end{enumerate}
\end{lemma}
\begin{proof}
Since the statements are of a local nature, we focus on a point $\bk_0$ and look at points $\bk$ which are sufficiently close to $\bk_0$. In this neighbourhood, a local frame for $P$ can be constructed. Indeed, let $\Phi(\bk_0)$ be any basis of $\Ran P(\bk_0)$. Then, for all $\bk$ close to $\bk_0$, the family
\begin{equation} \label{eq:Phik} 
\Phi(\bk) = P(\bk) \Phi(\bk_{0}) \left[ (P(\bk) \Phi(\bk_{0}))^{*} (P(\bk) \Phi(\bk_{0}))\right]^{-1/2}
\end{equation}
is well-defined, satisfies $\Phi^* \Phi = \bbI_n$ and $\Phi \Phi^* = P$, and hence is a (local) frame for $P$.
 
\begin{enumerate}[(i)]
\item  From \eqref{eq:Phik} we deduce that $\cF[P](\bk_0) = \cF[\Phi](\bk_0) =\rd \Tr \cA[\Phi](\bk_0)$, so that $\rd \cF[P](\bk_0) = \rd^2 \Tr \cA[\Phi](\bk_0) = 0$.  
\item Let $\Phi_P(\bk)$ (resp. $\Phi_Q(\bk)$) be defined as in~\eqref{eq:Phik}, so that it is a frame for $P(\bk)$ (resp. $Q(\bk)$) for $\bk$ sufficiently close to $\bk_0$. Since $P \perp Q$, the family $(\Phi_P, \Phi_Q)$ is a frame for $P+Q$. In particular, $\Tr \cA[(\Phi_P, \Phi_Q)] = \Tr \cA[\Phi_P] + \Tr \cA[\Phi_Q]$, and the result follows.
\end{enumerate}
\end{proof}

If $S$ is a closed compact surface, and if $P$ is smooth on $S$, then the \emph{Chern number} of $P$ on $S$ is
\[
\Ch(S,P) := \dfrac{1}{2\pi} \int_S \cF[P].
\]

Let $\Omega$ be a bounded connected open set in $\R^{3}$. From the Bianchi's identity together with Stokes theorem, any smooth family of projectors $P(\bk)$ defined on $\Omega$ satisfies $\int_{\partial \Omega} \cF[P] = \int_{\Omega} \rd \cF[P] = 0$, so that its Chern number $\Ch(\partial \Omega, P)$ vanishes. Therefore, a necessary condition for a projector defined on $\partial \Omega$ to have a smooth extension to $\Omega$ is to have a null Chern number. This condition is not only necessary but also sufficient, as we prove in the next Lemma, which is crucial to the proof of Theorem \ref{th:main}.

\begin{lemma} \label{lem:ChernOnS2}
Let $\Omega \subset \R^3$ be a bounded {connected} open set such that its boundary $\partial \Omega$ is a smooth, compact, oriented surface. Let $Q(\bk)$ be a smooth family of orthogonal projectors over $\Omega$, and let $P(\bK)$ be a smooth family of orthogonal projectors defined over $\partial \Omega$ such that $\Ran P \subset \Ran Q$ on $\partial \Omega$. The following propositions are equivalent:
\begin{enumerate}[(1)]
\item \label{item:Ch=0} the Chern number of $P$ on $\partial \Omega$ vanishes: $\Ch(\partial \Omega, P) = 0$;
\item \label{item:frame} there exists a smooth frame $\Phi$ over $\partial \Omega$ such that $P = \Phi \Phi^*$ on $\partial \Omega$;
\item \label{item:extension} there exists a smooth extension of $P$ on $\Omega$ such that $\Ran P \subset \Ran Q$ on $\Omega$.
\end{enumerate}
\end{lemma}
The proof of Lemma~\ref{lem:ChernOnS2} involves classical techniques from topology. We provide a short ``hands-on'' proof for completeness in the case of the three-dimensional ball $\Omega = \DD^3$, and sketch the modifications necessary to adapt this argument to the general case in Appendix~\ref{app:genus}.

\begin{proof}[Proof of Lemma~\ref{lem:ChernOnS2} for $\Omega=\DD^3$.] 
In the sequel we denote by $N_P$ the rank of $P$ and by $N_Q\geq N_P$ the rank of $Q$. We work in spherical coordinates: a point $\bk \in \DD^3$ is written as $\bk = r \bK$ with $r := | \bk | \in [0,1]$ and $\bK \in \bbS^2$.

\medskip

The fact that~\eqref{item:extension}$\implies$\eqref{item:Ch=0} is a consequence of Bianchi's identity together with Stokes' theorem, as was noted above. 

\medskip

\noindent {\bf Step 1: \eqref{item:Ch=0}$\implies$\eqref{item:frame}}.
We recall that if $P(s) \equiv P(\gamma(s))$ is a smooth family of projectors defined on a path $\gamma(s)$, and if $\Psi_0 = (\psi_1, \cdots, \psi_N)(s_0)$ is a frame for some $s_0$, then we can {\em parallel transport} $\Psi$ with $P$ along $\gamma$: the solution to
\begin{equation} \label{eq:parallelTransport}
\left\{ \begin{aligned}
& \frac{\rd \Psi}{\rd s}  = \left[ \frac{\rd P}{\rd s}, P \right] \Psi, \\
& \Psi(s_0)  = \Psi_0,
\end{aligned}
\right.
\end{equation}
satisfies $\frac{\rd(\Psi^* \Psi)}{\rd s} = 0$ and $\frac{\rd(\Psi^* P \Psi)}{\rd s} = 0$. In particular, $\Psi^* \Psi = \bbI_N$, \ie $\Psi(s) = (\psi_1, \cdots, \psi_M)(s)$ is a smooth frame. In addition, we have $\Psi(s)^* P(s) \Psi(s) = \Psi(s_0)^* P(s_0) \Psi(s_0)$, so that if $P(s_0) = \Psi(s_0) \Psi(s_0)^*$, then $P(s) = \Psi(s) \Psi(s)^*$ for all $s$.

\medskip

\noindent {\bf Step 1a: Construction of a continuous frame}.
Let $N_P$ be the rank of $P$. We divide the sphere $\bbS^2$ as $\bbS^2= \bbS^2_+ \cup \bbS^2_- \cup \cE$, where $\bbS^2_+$ is the northern hemisphere, $\bbS^2_-$ is the southern one, and $\cE$ is the (oriented) equator $\cE := \partial \bbS^2_+ = - \partial \bbS^2_-$. We fix two frames of $P$, one at the north pole, and the other one at the south pole, and we parallel transport these bases with the projector $P$ along the meridians of $\bbS^2$ (see~\eqref{eq:parallelTransport}). This process generates two bases $\Phi_{\pm}$ on $\bbS^{2}_{\pm}$. At the equator $\bK \in \cE$, both these bases span the range of $ P(\bK)$. We infer that there exists a unitary
$U_\obs(\bK) \in \U(N_P)$, defined on the equator $\bK \in \cE$, which is smooth and such that $\Phi_{+}(\bK) = \Phi_{-}(\bK) U_\obs(\bK)$.

We now attempt to contract continuously the loop $(U_\obs(\bK))_{\bK \in \cE}$ to the identity matrix in $\U(N_P)$. A necessary {(and actually, as we will prove, also sufficient)} condition for this to be possible is that the \emph{winding number} $W$ of its determinant vanishes. This winding number is given by
\[ W = \frac 1 {2\pi \ri} \int_{\cE} (\det U_\obs)^{-1} \rd \left( \det U_\obs \right) = \frac 1 {2\pi \ri}\int_{\cE} \Tr(U_\obs^{*} \rd U_\obs) \in \Z. \]
On the other hand, we can compute the connection of $\Phi_{+}$ on the equator, and get
\[ \cA[\Phi_+] = -\ri \Phi_{+}^{*} \rd \Phi_{+} = -\ri U_\obs^{*}\Phi_{-}^{*} (\Phi_{-} \rd U_\obs +\rd \Phi_{-} U_\obs) = U_\obs^* \cA[\Phi_-] U_\obs - \ri U_\obs^{*} \rd U_\obs. \]
Together with the Stokes theorem, this leads to%
\footnote{This calculation also proves that the Chern number, being equal to a winding number, is an integer.}%
\begin{equation}\label{hc101}
\begin{aligned}
W &= \frac 1 {2\pi} \int_{\cE} \Tr (\cA[\Phi_{+}] - \cA[\Phi_{-}]) = \frac 1 {2\pi} \left( \int_{\bbS^2_+} \cF[\Phi_+] + \int_{\bbS^2_-} \cF[\Phi_-] \right) = \frac{1}{2\pi} \int_{\bbS^2} \cF[P] \\
  &=\Ch(\bbS^2, P) = 0.
\end{aligned}
\end{equation}
We deduce that we can write $\det (U_{\obs}) = \exp(\ri \phi(\bK))$, where $\phi$ is continuous and periodic (compare \eg \cite[Lemma~2.13]{cornean2016construction}). Then $U_{\obs}(\bK) \exp(-\ri \phi(\bK)/N_{p})$ belongs to $\SU(N_{P})$ and can be continuously contracted to the identity%
\footnote{Explicitly deforming loops in $\SU(N_P)$ is not trivial. Several algorithms exist, such as the one-step logarithm method~\cite{fiorenza2016construction,cances2017robust}, the two-steps logarithm method~\cite{cornean2016construction, cornean2017_3D} and the recursive method~\cite{Gontier2017todo} for the general case.}%
; thus the same is true for $U_{\obs}(\bK)$. Parametrising such a contraction with respect to the latitude on the southern hemisphere we deduce that there exists a continuous map $\widetilde{U}$ on $\bbS^2_-$ which equals the identity at the south pole and satisfies the boundary condition $\widetilde{U} = U_\obs$ on $\cE$. The map $\Phi$ defined by $\Phi := \Phi_{+}$ on $\bbS^2_+$ and by $\Phi := \Phi_{-} \widetilde{U}$ on $\bbS^2_-$ is a continuous frame for $P$ on the whole sphere $\bbS^{2}$. 

\medskip

\noindent {\bf Step 1b: smoothing the frame on $\bbS^2$}. 
We extend (discontinuously) the basis to $\mathbb{R}^3$ by $\Phi(r\bK):=\Phi(\bK)$. Let $g\in C_0^\infty(\mathbb{R}^3)$,  $g\geq 0$, with $\int_{\mathbb{R}^3}g(\bk) \rd \bk=1$. Let $ \delta>0$ and define  $g_\delta(\bk)=\delta^{-3} g(\delta^{-1} \bk)$. If $\delta$ is small enough, the map $\Phi_\delta(\bK):=\int_{\mathbb{R}^3} g_\delta(\bK-\bk')\Phi(\bk') \rd \bk'$ is smooth and close in norm to $\Phi(\bK)$, uniformly in  $\bK \in \bbS^2$. Then the $N_P\times N_P$ overlap matrix  $G(\bK):=\Phi_\delta(\bK)^* P(\bK) \Phi_\delta(\bK)$ is uniformly close to $\bbI_{N_P}$ and the set of vectors  $[P(\bK)\Phi_\delta(\bK)]G(\bK)^{-1/2}$ forms the smooth frame we need. A similar argument has been used in \cite{cornean2016construction}.

\medskip

\noindent {\bf Step 2:~\eqref{item:frame}$\implies$\eqref{item:extension}}.
The idea is to construct a smooth extension for the frame $\Phi$ of $P$. 

\medskip

\noindent {\bf Step 2a: reduction to the case $Q(\bk) \equiv \mathrm{Id}$}.
Let $N_Q$ be the rank of $Q$, and let $\Psi(\bnull) = (\psi_1, \cdots, \psi_{N_Q})(\bnull)$ be a frame of $Q(\bnull)$. We parallel transport the frame $\Psi(\bnull)$ with the projector $Q$ along the rays of $\DD^3$ (see~\eqref{eq:parallelTransport}), and obtain a smooth family of frames $\Psi(\bk)$ for $Q(\bk)$ for all $\bk \in \DD^3$. Since $\Ran P \subset \Ran Q$ on $\bbS^2$, it holds that $P(\bK)$ is of the form $P(\bK) = \Psi(\bK) \widetilde{P}(\bK) \Psi(\bK)^{*}$, where $\widetilde{P}(\bK)$ is a smooth family over $\bbS^2$ of orthogonal projectors acting on $\widetilde{\cH} := \C^{N_Q}$. If $\Phi$ is a frame (in $\cH$) for $P$ on $\bbS^2$, then $\widetilde{\Phi} := \Phi \Psi^*$ is a frame (in $\widetilde{\cH}$) for $\widetilde{P}$ on $\bbS^2$. Conversely, if $\widetilde{\Phi}$ is extended from $\bbS^{2}$ to $\DD^{3}$, then $\Phi$ can be extended by $\Phi := \widetilde{\Phi} \Psi$. It is therefore enough to extend $\widetilde{\Phi}$ on $\DD^3$, with no $Q$-constraint. In the sequel, we drop the tildes, and assume that $Q(\bk) \equiv \mathrm{Id}$.
  
\medskip
  
\noindent {\bf Step 2b: construction of the extension}.
We build this extension inductively with respect to the rank $N_P$ of $P$.

If $N_P= 1$, then we can write $P(\bK) = \phi(\bK) \phi(\bK)^{*}$ for a smooth family of normalised vectors $\phi(\bK)$. If the dimension $M=1$, $\phi(\bK)$ can be chosen constant equal to $1$, and can trivially be extended to $\DD^{3}$. If  $M \geq 2$, the smooth map $\phi(\bK)$ for $\bK \in \bbS^{2}$ defines a two-dimensional surface in the unit sphere of $\cH$, which is isomorphic to $\bbS^{2M-1}$ and has dimension at least $3$. By Sard's lemma \cite[Chap.~1, \S7]{guillemin1974topology}, since $\phi(\bK)$ cannot fill the $2M-1$ dimensional sphere, there exists $-\phi^{*} \in \bbS^{2M-1}$ such that $\phi(\bK) \neq -\phi^{*}$ for all $\bK \in \bbS^{2}$. Let $0\leq h\leq 1$ be a $C^\infty(\R^+)$ cut-off function such that $h(r)=0$ for $r \le 1/4$ and $h(r)=1$ for $r \ge 3/4$. Then the vector 
\begin{equation} \label{eq:contractionTrick}
\widetilde{\phi}(r\bK) := \frac{h(r) \phi(\bK) + (1-h(r)) \phi^{*}}{|h(r) \phi(\bK) + (1-h(r)) \phi^{*}|}
\end{equation}
defines a smooth extension of $\phi$ to $\DD^{3}$. This proves the lemma when $N_P = 1$.

We now assume that the result is true for $N_P = n-1$, and we provide a construction for $N_P = n$. Let $(\phi_{1},\dots,\phi_{n})$ be a smooth frame for $P$ on $\bbS^{2}$. Using the previous construction for the case $N_P = 1$, we can smoothly extend $\phi_{1}$ to $\DD^{3}$. We now set $Q_1 := {\rm Id}_\cH - \phi_{1}\phi_{1}^{*}$ on $\DD^3$. From the reduction to the case $\Ran Q_1=\cH$ in {\bf Step 2a}, and the induction hypothesis, we can extend $P - \phi_{1}\phi_{1}^{*}$ to $\DD^{3}$ while keeping the extension contained inside $Q_1$. This provides the desired extension of $P$.
\end{proof}

We end this section with a result which will be used in conjunction with Lemma~\ref{lem:ChernOnS2} to prove Theorem \ref{th:main} in the TRS case. This time, we want to extend a projector defined on the circle to the disk. Here, no topological obstruction appears (vector bundles on the circle are trivial), but we seek an extension that preserves the TRS property.

\begin{lemma} \label{lem:TRS}
Let $Q(\bk)$ be a smooth family of orthogonal projectors over the two-dimensional disk $\DD^2$, which satisfies the TRS constraint $Q(-\bk) = \theta Q(\bk) \theta^{-1}$, and let $P(\bK)$ be a smooth family of orthogonal projectors over the circle $\bbS^1 := \partial \DD^2$, satisfying the TRS constraint $P(-\bK) = \theta P(\bK) \theta^{-1}$, and such that $\Ran P \subset \Ran Q$ on $\bbS^1$. Then,
\begin{enumerate}[(1)]
\item \label{item:frameonS1} there exists a smooth TRS frame $\Phi$ for $P$ on $\bbS^1$ (\ie $\Phi(-\bK) = \theta \Phi(\bK)$);
\item \label{item:extensiononD2} there exists a smooth TRS extension of $P$ on $\DD^2$ such that $\Ran P \subset \Ran Q$ on $\DD^2$.
\end{enumerate}
\end{lemma}
\begin{proof}
We first prove~\eqref{item:frameonS1}. To build a frame for $P$ on $\bbS^1$ we identify $\bbS^1$ with the unit circle of $\C$. We pick a basis at $\bK = 1$, and parallel transport it on the upper half circle $\bbS^1_+$ to get a family of bases $\widetilde{\Phi}$ on $\bbS^1_+$. Then  both $\widetilde{\Phi}(-1)$ and $\theta \widetilde{\Phi}(1)$ span the range of $P(-1)$, so that there exists $U_{\obs} \in \U(N_P)$ such that $\widetilde{\Phi}(-1) = \theta \widetilde{\Phi}(1) U_\obs$. The family $\Phi$ defined for $\bK = \exp(\ri \alpha) \in \bbS^1_{+}$ by $\Phi(\bK) := \widetilde{\Phi}(\bK) (U_{\obs})^{-\alpha/\pi}$ and extended to the lower half circle by $\Phi(-\bK) := \theta \Phi(\bK)$ is a continuous frame for $P$ on $\bbS^1$. We can then smooth it following {\bf Step 1b} of the proof of Lemma~\ref{lem:ChernOnS2}, which preserves the TRS property when the mollifier $g$ is even.

We now prove~\eqref{item:extensiononD2}, and first prove that we can choose $\Ran Q = \cH$, as in {\bf Step 2a}. We introduce the set of {\em real elements} of $\cH$ defined by
\[
\cH_{\rm r} := \left\{ v \in \cH: \ \theta v = v \right\}.
\]
{Let  $\Psi_Q(\bnull) = (v_1, \cdots, v_{N_Q})(\bnull)$ be a frame for $Q(\bnull)$.  Since $Q$ is TRS, it holds that ${Q(\bnull)\theta  = \theta Q(\bnull) }$ hence $ \theta \Psi_Q(\bnull)$ remains in the range of $Q(\bnull)$.  The vectors $v_j+\theta v_j$ and $\ri (v_j-\theta v_j)$ belong to  $\cH_{\rm r}$ and span the range of $Q(\bnull)$.  Using the Gram-Schmidt algorithm (which preserves ``reality'') we may find a {\em real} frame of $Q(\bnull)$, still denoted by $\Psi_Q$. Also, since the parallel transport~\eqref{eq:parallelTransport} preserves the TRS property (in the sense that the solution $\Psi_{Q}$ satisfies ${\Psi_Q(-\bk) \theta= \theta \Psi_Q(\bk)}$), we deduce as before that we can take $\Ran Q = \cH$ without loss of generality.

It remains to contract smoothly the vectors of the frame $\Phi$ on $\DD^2$, keeping the TRS constraint. We denote by $\sB_\cH$ the unit ball of $\cH$. Following the induction in {\bf Step 2b} of the proof of Lemma~\ref{lem:ChernOnS2}, it is enough to contract on $\sB_\cH$ a smooth family of vectors $(\phi(\bK))_{\bK \in \bbS^1}$, which satisfies $\phi(-\bK) = \theta \phi(\bK)$, in a TRS way. If $M = 1$, then $\phi(\bK)$ can be chosen constant and real, first in $\bbS^1$, then in $\DD^2$. 

If $M \geq 2$, then $\sB_\cH \cap \cH_{\rm r}$ is of real dimension $M-1$. We distinguish two cases. If $\phi(\bK)$ is never a real element for all $\bK \in \bbS^1$, then we can pick any real element $-\phi^* \in \sB_\cH \cap \cH_{\rm r}$, and construct a smooth extension like in~\eqref{eq:contractionTrick}, which also has the TRS property. 

It remains to study the case where there exists at least one $\bK_0 \in \bbS^1$ such that $\phi(\bK_0)$ is a real element. In this case, we have $\phi(-\bK_0) = \theta \phi(\bK_0) = \phi(\bK_0)$. Without loss of generality, we may assume that $\bK_0 = 1$. We draw the segment $[-\bK_0, \bK_0] = [-1, 1] \subset \DD^2$, and extend continuously $\phi$ on this segment with the constant value $\phi(1)$. Then we consider the decomposition $\DD^2 = \DD^2_+ \cup \DD^2_-$, where $\DD^2_+$ (resp. $\DD^2_-$) is the upper (resp. lower) half disk. In particular, $\DD^2_+$ is homeomorphic to $\DD^2$, and $\phi$ is well-defined, continuous, and piece-wise smooth on its boundary $\partial \DD^2_+$. Thanks to Sard's Lemma, and since $\sB_\cH$ is of dimension at least $2$, while $\phi(\bK)$ is a smooth curve (constantly degenerate on $[-1,1]$)}, there exists $-\phi^*$, which is not necessarily a real element, such that $-\phi^* \neq \phi(\bK)$ for all $\bK \in \partial \DD_+^2$. We can therefore extend continuously $\phi(\bK)$ from the boundary $\partial \DD^2_+$ to the whole upper half disk $\DD^2_+$, mimicking~\eqref{eq:contractionTrick} where the role of the ``origin'' is played by some arbitrary interior point of $\DD^2_+$. We then extend $\phi$ to the lower half disk $\DD^2_-$ by setting $\phi(-\bk) = \theta \phi(\bk)$ for $\bk \in \DD^2_+$. This extension is continuous and TRS.

We finally repair smoothness on $\DD^2$. We first convolve $\phi$ with a smooth and even approximation of the Dirac distribution to obtain a smooth TRS family $\phi_\delta(\bk)$ which is uniformly close in norm to $\phi(\bk)$ on $\DD^2$, and in particular on $\bbS^1$. Let $0<\varepsilon<1$ and consider a smooth even function $0\leq f_\varepsilon \leq 1$ such that $f_\varepsilon(\bk) = 1$ for $|\bk|\leq 1-2 \varepsilon$ while $f_\varepsilon(\bk) = 0$ for $| \bk | \geq 1 - \varepsilon$. Then for sufficiently small $\varepsilon$, the family of vectors 
$$ \phi_{\varepsilon,\delta}(r\bK):=\frac{ (1-f_\varepsilon(r\bK))\phi(\bK) +f_\varepsilon(r\bK)\phi_\delta(r\bK)}{|(1-f_\epsilon(r\bK))\phi(\bK) +f_\epsilon(r\bK)\phi_\delta(r\bK)|}$$
is the smooth TRS extension we need.
\end{proof}

\section{Proof of Theorem~\ref{th:main}} \label{sec:proof}

We now provide the proof of Theorem~\ref{th:main}. We distinguish the cases without and with TRS in two subsections.

\subsection{Without the TRS constraint} \label{ssec:proof_noTRS}
We construct the projector ${P}$ on the Brillouin zone $\cB$ {\em piece by piece}, and glue the pieces in a smooth way. First note that, when $K_{N+1} = \emptyset$, \ie when there are no crossings between the $(N+1)$-st and $(N+2)$-nd bands, Theorem~\ref{th:main} is trivial: simply take $P = P_{N+1}$, which is smooth and whose range contains that of $P_{N}$. Otherwise, our strategy of proof will be to take $P$ equal to $P_{N+1}$ on most of the Brillouin zone except for a small set $\Omega$, and extend it in a smooth way to $\Omega$ using Lemma \ref{lem:ChernOnS2}.

The natural idea would be to take $P$ equal to $P_{N+1}$ everywhere, except in small neighbourhoods near the crossing set $K_{N+1}$, and continue it inside those neighbourhoods. The main difficulty we face is the local topology of the crossings in $K_{N+1}$ (see \eg \cite{monaco2014graphene} and references therein). For concreteness, assume that $\cH$ is of dimension 2 and that $H$ has crossings on a set $K$ composed of isolated points. Assume we want to extend the spectral projector $P$ on the subspace associated to the first eigenvalue of $H$, defined on a neighborhood of $\BZ \setminus K$, to the whole Brillouin zone $\BZ$. Since $\cH$ is of dimension 2, the Hamiltonian can be written as $H(\bk) = V(\bk) \bbI_2 + B(\bk) \cdot \sigma$, where $V(\bk)$ is a scalar, $B(\bk)$ is a three-dimensional vector field, and $\sigma$ denotes the vector of the three usual Pauli matrices. Then $B$ has isolated zeroes on the points of $K$, and it can be checked that the Chern number of $P$ on a small sphere $S$ around a point $\bk_{0} \in K$ is equal to the index of $B$ at that point (defined as the degree of $B(\bk)/|B(\bk)|$ as a map from $S$ to the unit sphere). In particular, if $B'(\bk_{0})$ is non-degenerate, then $\Ch(P,S) = \frac{\det B'(\bk_{0})}{|\det B'(\bk_{0})|} = \pm 1$. Such points are called \textit{Weyl points}, and $\Ch(P,S)$ is called the \textit{Weyl charge}. Because $P$ has a non-zero Chern number at those points, it is impossible to continue $P$ to the interior of $S$, in view of  Lemma \ref{lem:ChernOnS2}.

This example shows that it is in general impossible to extend $P_{N+1}$ in a neighbourhood of a Weyl point, for topological reasons. However, the example above provides the way out. Indeed, the Poincar\'e--Hopf theorem \cite[Chap.~3, \S~5]{guillemin1974topology} states that the sum of the indices of $B$ is equal to the Euler characteristic of $\BZ$, which is zero. In other words, the sum of the Weyl charges is zero, a statement sometimes referred to as the Nielsen--Ninomiya Theorem~\cite{nielsen1981no,friedan1982proof,mathai2017global}. Therefore, while it is \textit{locally} impossible to extend $P$,
there is no \textit{global} topological obstruction%
\footnote{Note that the time-reversal symmetry does not help here, since TRS pairs of Weyl points have the same Weyl charge.}%
. Our strategy of proof is therefore to group the Weyl points together in an open set $\Omega$ which does not intersect with $K_{N}$, to define $P$ equal to $P_{N+1}$ outside this set, and to extend it inside by Lemma \ref{lem:ChernOnS2}, while imposing $\Ran P_{N} \subset \Ran P$.

To construct the set $\Omega$, we use Assumption 1. Let $K_{N+1}^{\mathrm{pt}}$ be the set of isolated points in $K_{N+1}$. Since the sets $K_N$ and $K_{N+1}$ are disjoint and at most 1-dimensional, there exists a point $\bk^* \in \BZ$ such that the segments joining $\bk^*$ with the points of $K_{N+1}^{\mathrm{pt}}$ do not touch the points of $K_{N}$. The union of these segments can be then joined to each of the curves that constitute $K_{N+1} \setminus K_{N+1}^{\mathrm{pt}}$, again without crossing $K_N$. We then define $\Omega$ to be a small neighbourhood of the union of these segments, with smooth boundary, so that $\Omega \cap K_{N} = \emptyset$ (see Figure~\ref{fig:omega}).

\begin{figure}[H]
\centering
\begin{tikzpicture}
\draw (-2.5,-2.5) -- (2.5,-2.5) -- (2.5,2.5) -- (-2.5, 2.5) -- (-2.5, -2.5); \node at (0, 2.7) {$\BZ$};
\foreach \point in {(-1, -2), (0,0), (1, -2), (1.5, 1.5)} \fill[blue] \point rectangle +(0.15, 0.15);
\draw [blue,line width=2pt] (2,-1) rectangle +(-0.5,-.5);
\foreach \point in {(1,1), (2,1), (-1, 2), (-2, -2)}	{
\draw[red!40, line width=10pt, line cap=round] (0,-1) -- \point;
\fill[black] \point circle (2pt);
}
\draw [red!40, line width=10pt] (-2.5,.5) -- (2.5,.5);
\draw [line width=2pt] (-2.5,.5) -- (2.5,.5);
\draw [red!40, line width=10pt] (.5,1.75) ellipse (20pt and 10pt)
      (0,.6) -- (0,1.5);
\draw [line width=2pt] (0.5,1.75) ellipse (20pt and 10pt);
\node[draw, line width=2, black, cross out]  at (0,-1) {};
\node at (0,-1.4) {$\bk^*$};
\fill[black] (3, -2) circle (2pt); \node at (3.7, -2) {$K_{N+1}$};
\fill[blue] (3,-1.5) rectangle +(0.15, 0.15); \node at (3.7, -1.5) {$K_N$};
\node[red] at (1, -0.5) {$\Omega$};
\end{tikzpicture}
\caption{2D sketch of the set $\Omega \subset \BZ$ enclosing the set $K_{N+1}$ but avoiding the set $K_N$.}
\label{fig:omega}
\end{figure}
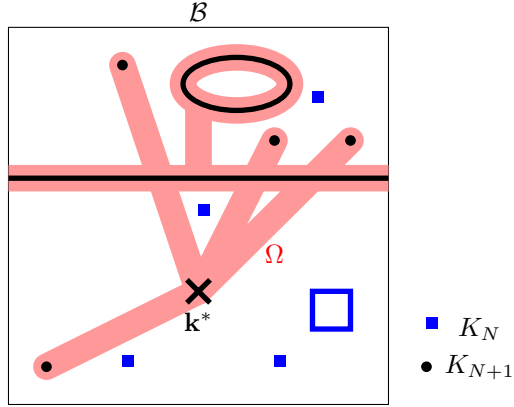

Since $K_{N+1} \subset \Omega$, the projector $P_{N+1}$ is well-defined and smooth on $\BZ \setminus \overline{\Omega}$. On the other hand, since $\Omega \cap K_N  = \emptyset$, the projector $P_N$ is well-defined and smooth on $\Omega$, and is moreover orthogonal to $P_{N+1}$ on $\partial \Omega$. We define, for $\bk \in \partial \Omega$, the rank-$1$ projector $p(\bk) := P_{N+1}(\bk) - P_N(\bk)$. Our goal is to extend $p$, initially defined on the boundary $\partial \Omega$, to the interior $\Omega$, while keeping the constraint $ p \in \Ran P_N^\perp$. Using Lemma~\ref{lem:ChernOnS2}, it is enough to prove that $\Ch(\partial \Omega, p) = 0$. From Lemma~\ref{lem:sumBerry}, we have that
\[
\Ch(\partial \Omega, p) = \Ch(\partial \Omega, P_{N+1}) - \Ch(\partial \Omega, P_N).
\]
But $P_{N}$ is well-defined on $\Omega$, and therefore 
\[
\Ch(\partial \Omega, P_{N}) = \frac 1 {2\pi} \int_{\partial \Omega} \cF[P_{N}] = \frac 1 {2\pi} \int_{\Omega} \rd \cF[P_{N}] = 0.
\]
Similarly, $P_{N+1}$ is well-defined on $\BZ \setminus \Omega$, which has boundary $-\partial \Omega$, and so
\[
\Ch(\partial \Omega, P_{N+1}) = \frac 1 {2\pi} \int_{\partial \Omega} \cF[P_{N+1}] = -\frac 1 {2\pi} \int_{\BZ \setminus \Omega} \rd \cF[P_{N+1}] = 0.
\]
This proves that the total sum of the Weyl charges $\Ch(\partial \Omega, p)$ must vanish, which can be seen as a generalisation of the Nielsen--Niyomiya theorem to the case of line degeneracies.

Altogether, we proved that $\Ch(\partial \Omega, p) = 0$. Together with Lemma~\ref{lem:ChernOnS2}, we can smoothly extend $p$ over $\Omega$ while keeping the constraint $\Ran p \subset \Ran P_{N}^{\perp}$. In conclusion, the map
\[
{P}_{\rm cont} := \begin{cases}
P_{N+1} & \text{on } \BZ \setminus \Omega, \\
P_{N} + p & \text{on } \Omega.
\end{cases} 
\]
is well defined and continuous over $\BZ$. It is smooth on $\Omega$ and on $\BZ \setminus \Omega$, but not necessarily across the boundary $\partial\Omega$. In order to find a smooth extension, we interpolate between the two pieces. 

More specifically, for $\varepsilon > 0$ we define 
\begin{equation} \label{eq:def:Omegaepsilon}
\Omega_\varepsilon := \left\{ \bk \in \Omega, \ \dist(\bk, \partial \Omega) > \varepsilon \right\}.
\end{equation}
When $\varepsilon$ is small enough, it holds that $\overline{\Omega_\varepsilon} \subset \Omega$ and $K_{N+1} \subset \Omega_\varepsilon$. In this case, the operator $P_{N+1}$ is well-defined and smooth on $\BZ \setminus \Omega_\varepsilon$. Let $f_\varepsilon: \BZ \to [0,1]$ be a smooth cut-off function such that $f_\varepsilon(\bk) = 0$ on $\Omega_\varepsilon$ and $f_\varepsilon(\bk) = 1$ on $\BZ \setminus \Omega$, and let
\[
\widetilde{P}_{\varepsilon}(\bk) :=
f_\varepsilon(\bk) P_{N+1}(\bk) + (1 - f_\varepsilon(\bk)) \left( P_N(\bk) + p(\bk) \right).
\]
The map $\bk \mapsto \widetilde{P}_{\varepsilon}(\bk)$ is a smooth family of self-adjoint operators, which is uniformly close in norm to ${P}_{\rm cont}$, and whose range contains the one of $P_N$. In particular, for $\varepsilon$ small enough, $\widetilde{P}_{\varepsilon}$ has exactly $N+1$ eigenvalues above $3/4$, and the remaining eigenvalues below $1/4$. We finally take
\[
{P} := \dfrac{1}{2 \pi \ri} \oint_{\sC} \left( z - \widetilde{P}_{\varepsilon} \right)^{-1} \rd z, 
\]
where $\sC$ is the circle with center $1$, and radius $1/2$. This corresponds to diagonalising {$\widetilde{P}_{\varepsilon}$} and setting the eigenvalues less than $1/2$ to $0$ and those greater than $1/2$ to $1$, and so $P$ is a projector. {Since $(z- widetilde{P}_{\varepsilon})P_N=(z-1)P_N$ we conclude that $PP_N=P_N$. Thus the projection valued map ${P}$ is smooth on $\BZ$ and its range contains the range of $P_N$}. This concludes the first part of the proof of Theorem~\ref{th:main} in the non-TRS case.

When Assumption 2 is also satisfied, then $\Omega$ can be chosen so that it also avoids $K_{N+2}$. In this case, the projector $P_{N+2}$ is well-defined on $\Omega$, and we can look for the extension of the rank-$1$ projector $p$ in $\Ran (P_{N+2} - P_{N})$. Such an extension leads to a smooth family of projectors $P$ on $\Omega$ such that $\Ran P_N \subset P \subset \Ran P_{N+2}$.

\subsection{With the TRS constraint}

We now treat the case where $H$ is TRS. If we could split $\Omega$ into disjoint parts $\Omega_{\pm}$ related by the time-reversal
symmetry operation $S(\bk) = -\bk$ through $S(\Omega^{\pm}) = \Omega^{\mp}$, then we could apply our previous construction to find $P$ on $\Omega_{+}$, and deduce $P$ on $\Omega_{-}$ by symmetry. Unfortunately, this is complicated by possible crossings at time-reversal invariant momenta (TRIMs). Recall that TRIMs are defined as the fixed points of the map $S$ in $\BZ$. There are eight distinct TRIMs in three dimensions, namely $\{0, 1/2\}^3$. We denote by $K_\TRIM$ the TRIM of $K_{N+1}$. If $K_{\TRIM}$ is not empty, then one cannot  find a disjoint TRS pair $\Omega^{\pm}$ containing $K_{N+1}$. Another complication is crossings along curves that pass through the Brillouin zone (an example is the crossing of ${\rm sp}^3$ electrons along the $\Gamma \to X \to \Gamma$ line in semiconductors in zinc-blende structures such as Silicon \cite{cohen1966band}). To solve this problem, we will cut $\Omega$ along TRS-invariant surfaces.

Let $\cC$ be a set of connected smooth curves that contains $K_{N+1}$, avoids $K_N$ (and $K_{N+2}$ if required), and such that $\cC$ is TRS-invariant, in the sense that $S(\cC) = \cC$. This is always possible thanks to Assumption 1. We choose two disjoint smooth surfaces $\Sigma_{1,2} \subset \BZ$ that cut the Brillouin zone in two disjoint components $\BZ_{\pm}$ such that $\BZ_{\pm} = S(\BZ_{\mp})$, and that cross $\cC$ transversally, so that $\cC \cap \{\Sigma_{1} \cup \Sigma_2\}$ is a union of finitely many points $(d_{j})_{1\leq j \leq J}$, among which are the points in $K_{\TRIM}$ (see Figure~\ref{fig:pbSym}). Essentially, $\Sigma_{1,2}$ are the two planes $\{x_3 = 0\}$ and $\{x_3 = 1/2\}$, that are deformed so that they cross $\cC$ at a finite number of points. Let $\Omega$ be a neighbourhood of $\cC$ with smooth boundary such that $S(\Omega) = \Omega$. When the neighbourhood is sufficiently close to $\cC$, $\Omega$ does not intersect $K_{N}$ nor $K_{N+2}$, and $\Omega \cap \{\Sigma_{1} \cup \Sigma_2\}$ is the disjoint union of $J$ sets $\cD_{j}$, with each $\cD_{j}$ diffeomorphic to the disk $\DD^{2}$. Setting $\Omega_{\pm} := \Omega \cap \BZ_{\pm}$, the set $\Omega$ is split in the disjoint union $\Omega = \Omega_{+} \cup \Omega_{-} \cup (\bigcup_{j=1}^{J} \cD_{j})$, and it holds that $\Omega_{\pm} = S(\Omega_{\mp})$.

\begin{figure}[H]
	\centering
	\begin{subfigure}[b]{0.4\textwidth}
	\begin{tikzpicture}
	\newcommand\Trim{+(-0.07,-0.07) rectangle +(0.07,0.07)}
	\foreach \point in {(0,0), (-2.5, 2.5), (-2.5, 0.5), (-1,1.8)}  {
		\draw[red!40, line width=5pt, line cap=round] (-1.5,1) -- \point;}
	\foreach \point in {(0,0), (2.5, -2.5), (2.5, -0.5), (1,-1.8)}  {
		\draw[red!40, line width=5pt, line cap=round] (1.5,-1) -- \point;}
	
	\draw [red!40, line width=5pt] (0,0) ellipse (30pt and 10pt);
	\draw [line width=2pt] (0,0) ellipse (30pt and 10pt);
	\draw[domain=-2.5:2.5,smooth,variable=\x,black!70, line width=1.5pt] plot ({\x},{0.5*sin(deg(\x*pi/2.5))});
	\node[black!70] at (1.4, 0.7) {$\Sigma_1$};
	\draw[domain=-2.5:2.5,smooth,variable=\x,black!70, line width=1.5pt] plot ({\x},{2.5-0.2*sin(deg(\x*pi/2.5))});
	\draw[domain=-2.5:2.5,smooth,variable=\x,black!70, line width=1.5pt] plot ({\x},{-2.5-0.2*sin(deg(\x*pi/2.5))});
	\node[black!70] at (1.4, 2) {$\Sigma_2$};
	\draw[white, line width=11pt] (-2.7,-2.7) -- (2.7,-2.7) -- (2.7,2.7) -- (-2.7, 2.7) -- (-2.7, -2.7);
	\foreach \point in {(-2.5, 0.5), (-1,1.8), (2.5, -0.5), (1,-1.8)}  {
		\fill[black] \point circle (2pt);}
	\foreach \point in {(0,0), (2.5, -2.5),  (-2.5, 2.5)}  {
		\fill[blue] \point \Trim;}
	\fill[black] (3, -1.5) circle (2pt); \node at (3.1, -1.5) [anchor=west] {$K_{N+1}$};
	\fill[blue] (3, -2) \Trim; \node at (3.1, -2) [anchor=west] {$K_{\TRIM}$};
	\node[red] at (-0.7, 0.8) {$\Omega$};
	\draw (-2.5,-2.5) -- (2.5,-2.5) -- (2.5,2.5) -- (-2.5, 2.5) -- (-2.5, -2.5); \node at (0, 2.7) {$\BZ$};
    \draw [red, thick] (.54,.29) circle (4pt)
                       (-.54,-.29) circle (4pt)
                       (3, -1) circle (4pt);
    \node at (3.1, -1) [anchor=west] {$\{d_j\} \setminus K_{\TRIM}$};
	\end{tikzpicture}
		\caption{2D sketch of the Brillouin zone.}
	\label{fig:pbSym}
	\end{subfigure}
	\hfill
	\begin{subfigure}[b]{0.4\textwidth}
	\includegraphics[scale = 1]{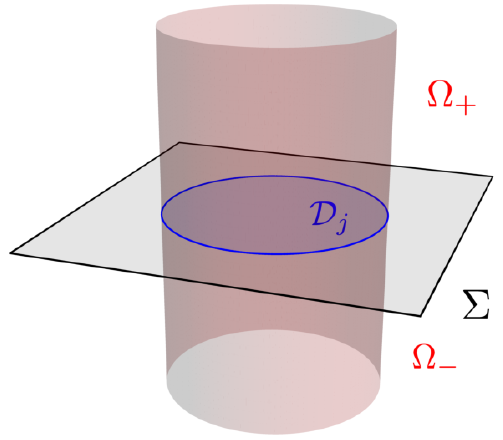}
	\vspace{1cm}
	\caption{The cut of $\Omega$ by the surfaces $\Sigma_{1,2}$.}
	\label{fig:cutOmega}
	\end{subfigure}
\caption{\eqref{fig:pbSym} The symmetric set $\Omega \subset \BZ$ enclosing $K_{N+1}$. Here, $K_{N+1}$ contains the two TRIM $(0,0)$ and $(-1/2, 1/2)$. The blue curves are the surfaces $\Sigma_{1,2}$ cutting the Brillouin zone $\BZ$ in two symmetric pieces. \eqref{fig:cutOmega} 3D visualisation of one of the cuts $\cD_j$.}
\end{figure}
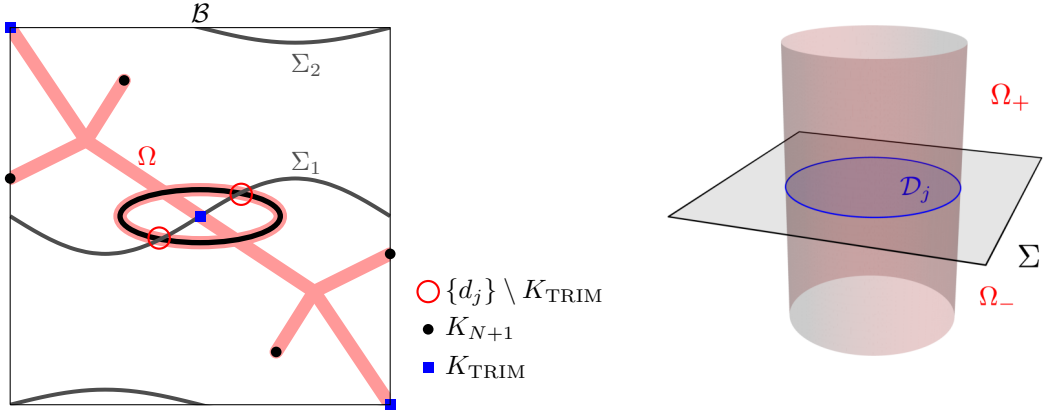

We now extend $P(\bk)$, which is smooth and well-defined on $\partial \Omega$, on each of the cuts $\cD_j$, while preserving TRS. For all $1 \leq j \leq J$, the family of projectors $P_N$ is well-defined and smooth on $\cD_{j}$, and the family of rank-$1$ projectors $p_{j} := P_{N+1} - P_N$ is well-defined and smooth on $\partial \cD_j$. By construction $S(\Sigma_{i}) = \Sigma_{i}$, and for each $1 \leq j \leq J$, there are two cases: either $d_{j}$ is a TRIM, in which case $S(\cD_{j}) = \cD_{j}$, or it is not, in which case $S(\cD_{j}) = \cD_{j'}$ with $j' \neq j$. In the first case, we use Lemma~\ref{lem:TRS}  to extend smoothly $p_j$ on $\cD_j$ such that its range remains inside $\Ran (P_{N+2} - P_{N})$, while preserving TRS. In the second case, we can perform the extension by following the lines of Lemma~\ref{lem:ChernOnS2} (this is a two-dimensional extension problem, and so there is no topological obstruction in this case). The extension on $\cD_j$ induces the one on $\cD_j'$ by $P(\bk) = \theta P(S(\bk)) \theta^{-1}$ for $\bk \in \cD_j'$.

Doing so for each cut $\cD_j$, we end up with a continuous and piecewise smooth family of rank-$(N+1)$ orthogonal projectors $P$ defined on $(\BZ \setminus \Omega) \cup \cD$ by  
\begin{equation} \label{eq:P_TRS}
P_{\rm cut} := \begin{cases}
   P_{N+1} & \quad \text{on} \quad \BZ \setminus \Omega, \\
   P_N + p_j & \quad \text{on} \quad  \cD_{j} \quad \text{for} \quad 1 \le j \le J.
\end{cases}
\end{equation}

It remains to extend this construction to the whole set $\Omega$. Recall that the cuts $\left( \cD_{j} \right)$ separate $\Omega$ in two different connected pieces $\Omega_{\pm}$ (see Figure~\ref{fig:omegawithcuts}).

\begin{figure}[H]
\centering
\begin{tikzpicture}[x=0.5cm,y=0.5cm]
\foreach \deltax in {-2.5,0,2.5}{
\foreach \deltay in {-2.5,0,2.5}{
\begin{scope}[xshift=\deltax cm, yshift=\deltay cm]
\foreach \point in {(0,0), (-2.5, 2.5), (-2.5, 0.5), (-1,1.8)}  {
\draw[red!40, line width=4pt, line cap=round] (-1.5,1) -- \point;
\fill[black] \point circle (1pt);
}
\foreach \point in {(0,0), (2.5, -2.5), (2.5, -0.5), (1,-1.8)}  {
\draw[red!40, line width=4pt, line cap=round] (1.5,-1) -- \point;
\fill[black] \point circle (1pt);
}
\draw [red!40, line width=3pt] (0,0) ellipse (20pt and 10pt);
\draw [line width=1pt] (0,0) ellipse (20pt and 10pt);

\draw (-2.5,-2.5) -- (2.5,-2.5) -- (2.5,2.5) -- (-2.5, 2.5) -- (-2.5, -2.5);
\draw[blue, line width=1.5] (-0.2, -0.2) -- (0.2, 0.2);
\draw[blue, line width=1.5] (-2.7, 2.3) -- (-2.3, 2.7);
\draw[blue, line width=1.5] (2.7, -2.3) -- (2.3, -2.7);
\draw[blue, line width=1.5] (-1, -0.7) -- (-0.6, -0.3);
\draw[blue, line width=1.5] (1, 0.7) -- (0.6, 0.3);
\end{scope}
}}
\node[red] at (1,1.3) {$\Omega$};
\end{tikzpicture}
\caption{The Brillouin zone $\BZ$ and the set $\Omega$ seen as periodic sets. The blue segments mark the cuts at the intersections $\cD_j$.}
\label{fig:omegawithcuts}
\end{figure}
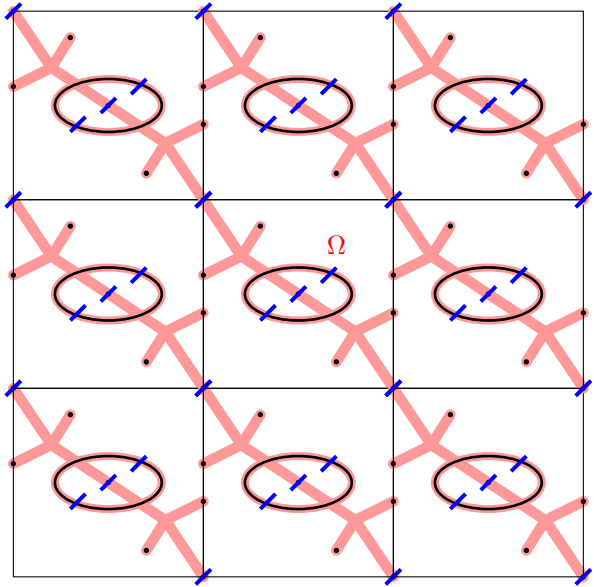

The map $P_{\rm cut}$ in~\eqref{eq:P_TRS} is well-defined, continuous and piece-wise smooth on the boundary $\partial \Omega_+$. Let us prove that we can extend $P_{\rm cut}$ continuously on $\Omega_+$. We set $p := P_{\rm cut} - P_N$ on $\partial \Omega_+ \cup \partial \Omega_-$, so that $\Ch(\partial \Omega_\pm, P_{\rm cut}) = \Ch(\partial \Omega_\pm, p) + \Ch(\partial \Omega_\pm, P_N)$. Since $P_N$ is well-defined on $\Omega$ and $P_{N+1}$ is well-defined on $\BZ \setminus \Omega$, we have that $\Ch(\partial \Omega_\pm, P_N) = \Ch(\partial \Omega_{\pm}, P_{N+1}) = 0$, and therefore
\[
0 = \Ch(\partial \Omega, p) = \Ch( \partial \Omega_+, p) + \Ch(\partial \Omega_-, p).
\]
On the other hand, from the TRS property, it holds that $\cF[p](S(\bk)) = - \cF[p](\bk)$ (see \eg \cite{panati2007triviality}). In particular,
\[
\Ch(\partial \Omega_+, p) = \dfrac{1}{2 \pi} \int_{\partial \Omega_+} \cF[p] = \dfrac{1}{2 \pi} \int_{S(\partial \Omega)} \left( - \cF[p]\right) = \dfrac{1}{2 \pi} \int_{\partial \Omega_-}  \cF[p] = \Ch(\partial \Omega_-, p),
\]
where we used the fact that $S (\partial \Omega_+) = -\partial \Omega_-$: the TRS reverses orientations. Altogether, this proves that $\Ch(\partial \Omega_+, p)  = \Ch(\partial \Omega_+, P_{\rm cut}) = 0$. From Lemma~\ref{lem:ChernOnS2}, we deduce that\footnote{Actually, since $P_{\rm cut}$ and $\partial \Omega_\pm$ are only continuous, the Chern number is not well-defined a priori as the integral of a differential form. Still, since everything is continuous and piece-wise smooth, we can easily adapt the proof of Lemma~\ref{lem:ChernOnS2} to handle this particular case.} we can construct a continuous extension of $P_{\rm cut}$ to $\Omega_+$, and then on $\Omega$ by setting $P_{\rm cont}(\bk) = \theta P_{\rm cut}(\bk)\theta^{-1}$ for $\bk \in \Omega_-$.

It remains to smooth this family out. To do so, let $g \in C^\infty_0(\R^3)$ be even with $g \ge 0$ and $\int_{\R^3} g(\bk) \rd \bk = 1$. We set $g_\delta(\bk) := \delta^{-3} g(\delta^{-1} \bk)$ for $\delta>0$. If $\delta$ is small enough, the family of self-adjoint operators $p_\delta(\bk) := \int_{\R^3} (P_{\rm cont}(\bk') -P_N(\bk'))g(\bk - \bk') \rd \bk'$ is well-defined and smooth near a neighbourhood of $\Omega$, and also TRS. If $\delta$ is small enough, the operator $A_{N,\delta}(\bk):=P_N(\bk)+({\rm Id}-P_N(\bk))p_\delta(\bk) ({\rm Id}-P_N(\bk))$ obeys $A_{N,\delta}P_N=P_N$, it is also smooth and uniformly close in norm to $P_{N+1}$ near $\partial \Omega$. In particular, we can smoothly interpolate between $A_{N,\delta}$ defined on $\Omega$, and $P_{N+1}$ defined outside $\Omega_\varepsilon \subset \Omega$, while keeping the TRS constraint; see~\eqref{eq:def:Omegaepsilon} and the following construction. This concludes the proof of Theorem~\ref{th:main}.

\appendix

\section{Analyticity is impossible} \label{app:analyticity}

Here we illustrate how the regularity of the family of projections $P$ claimed in our main Theorem~\ref{th:main} is the best possible, and cannot be in general pushed to analyticity due to the presence of Weyl points, as discussed in Section~\ref{sec:discuss}.
\begin{lemma} \label{lem:analyticity} 
If $K_N$ contains at least one Weyl point, then the only analytic projector of rank $N+1$ which spans $P_N$ in a neighbourhood of $K_N$ is $P_{N+1}$.
\end{lemma}
In particular, if $K_{N+1}$ is non empty, then there is no analytic projector of rank $N+1$ containing $P_{N}$ on $\BZ \setminus K_{N}$. This means that Wannier functions for metals can not be exponentially localised in general. We do not know whether we can always find an analytic rank-$(N+2)$ projector that contains $P_N$ in the case where $K_{N+2} \neq \emptyset$.

\begin{proof}[Proof of Lemma~\ref{lem:analyticity}.]
Assume that $H(\bk)$ {is real analytic and} has a Weyl point (a conical crossing of two eigenvalues) at a point $\bk_{0}$. We assume without loss of generality that the crossing happens between the first two eigenvalues of $H(\bk_{0})$ at level $0$. Then, in any direction ${\bf d} \in \R^{3}$, the function $k \mapsto H(\bk_{0} + k{\bf{d}})$ can be partially diagonalised as
\begin{align*}
  H(\bk_{0} + k{\bf{d}})=
  U(k)
\left( 
  \begin{array}{cc|c}
{c} \,k {+\mathcal{O}(k^2)}&0&0\\
0&-{c} \, k {+\mathcal{O}(k^2)}&0\\\hline
0&0&M(k)
\end{array} \right)
U(k)^{-1}
\end{align*}
where $c>0$ is constant, $M$ is an {operator whose spectrum is bounded from below by a positive constant, and $U(k)$ is an analytic family of unitary operators} {(see \eg \cite[Theorem XII.13]{reed1978analysis})}. We would like to construct a projector $P(k)$ of rank 2 that is analytic with respect to $k$ and {whose range includes} the eigenspace corresponding to the first eigenvalue of $H$ in a neighbourhood of $k=0$. Assume this is possible. Then $\widetilde{P}(k) = U(k)^{-1}P(k) U(k)$ is analytic and $\widetilde{P}(k) = \widetilde{P}_{0} + \sum_{j=1}^{\infty} \widetilde{P}_{j} k^{j}$ in a neighbourhood of $0$. The projector $\widetilde{P}$ must satisfy $\widetilde P(k) e_{1} = e_{1}$ for $k < 0$, and $\widetilde P(k) e_{2} = e_{2}$ for $k > 0$. By continuity,
\begin{align*}
  \widetilde{P}_{0} =\left( 
  \begin{array}{c|c}
    \bbI_{2}&0\\\hline
    0&0
  \end{array} \right).
\end{align*}

Let $n$ be the smallest index $j$ such that $\widetilde{P}_{j} \neq 0$. Then for all $k<0$ we must have that 
\[
\widetilde{P}(k)e_1=e_1= \widetilde{P}(0)e_1+k^n \widetilde{P}_n e_1 +\mathcal{O}(k^{n+1}),
\]
hence $\widetilde{P}_{n} e_{1} =0$. Similarly, $ \widetilde{P}_{n} e_{2} = 0$, and therefore $\widetilde{P}_{n} \widetilde{P}(0)= \widetilde{P}(0) \widetilde{P}_{n} = 0$, showing that the off-diagonal terms in $\widetilde{P}_{n}$ are zero. Also, from {$\widetilde{P}(k)^2=\widetilde{P}(k)$} we get $\widetilde{P}_{n} = \widetilde{P}(0)\widetilde{P}_{n} + \widetilde{P}_{n}\widetilde{P}(0)$ hence $\widetilde{P}_{n}=0$, contradicting our assumption on $n$. This leads to $\widetilde{P}(k) =\widetilde{P}(0)$ and therefore $P(k) = U(k)\widetilde{P}(0)U(k)^{-1}$ for $k$ small enough. Since ${\bf d}$ was arbitrary and by analyticity, it follows that $P(\bk)$ is identically equal to the spectral projection of $H(\bk)$ corresponding to its two lowest eigenvalues. The statement follows.
\end{proof}

\section{Extension of projections on general surfaces} \label{app:genus}

In this Appendix, we indicate how to modify the proof of Lemma~\ref{lem:ChernOnS2} to the case of a general bounded open set $\Omega \subset \R^3$ with smooth compact boundary $\partial \Omega$. Once again, \eqref{item:extension}$\Longrightarrow$\eqref{item:Ch=0} follows from the Bianchi's identity together with Stokes' theorem. Hence we just need to show that \eqref{item:Ch=0}$\Longrightarrow$\eqref{item:frame}$\Longrightarrow$\eqref{item:extension}.

According to the classification of closed orientable surfaces, $\partial \Omega$ is diffeomorphic to a surface of genus $g$. For instance, the case $\partial\Omega= \bbS^2$ treated in Section~\ref{sec:chern} corresponds to $g=0$, the case of a torus corresponds to $g=1$, and so on. Let us prove \eqref{item:Ch=0}$\Longrightarrow$\eqref{item:frame}. Although this can can be argued by means of standard techniques from algebraic topology (see \eg~\cite[Prop.~4]{panati2007triviality} and references therein), we provide here an explicit construction of a smooth frame (see also~\cite{monaco2017chern}). One consider the fundamental polygon of $\partial \Omega$, obtained by cutting the surface $\partial \Omega$ along the $2g$ loops $\{a_i, b_i\}_{1\le i \le g}$, generating its fundamental group as cycles~\cite[\S2.4]{jost2013compact}. One obtains a $4g$-gon $F$ whose sides are labelled by $a_1$, $b_1$, $a_1^{-1}$, $b_1^{-1}$,\ldots, $a_g$, $b_g$, $a_g^{-1}$, $b_g^{-1}$, in this order (see Figure~\ref{fig:genus}). By means of parallel transport and holonomy cancellation \cite{cornean2016construction}, one can construct a frame on the sides labelled by $a_i$ and $b_i$, which are then extended to the sides labelled by $a_i^{-1}$ and $b_i^{-1}$ by the obvious identifications. We obtain a frame $\Phi_{\partial F}$ on $\partial F$ which satisfies the periodicity conditions. On the other hand, we can start from a frame at the center of $F$, and parallel transport it to the whole cell $F$. We obtain a frame $\Phi_F$ defined on $F$, which may not satisfy the periodicity conditions at the boundary. The ``mismatch'' between $\Phi_F$ and $\Phi_{\partial F}$ on $\partial F$ is encoded in an obstruction matrix $U_{\rm obs}$ defined on $\partial F \simeq \bbS^1$, and whose winding number coincides with the Chern number of the family of projections along $\partial \Omega$ (compare \eqref{hc101}). The vanishing of the latter allows to extend $U_{\rm obs}$ on $F$, and therefore to cure $\Phi_F$. Notice that this construction was already indicated in \cite{monaco2015symmetry} (see also \cite{fiorenza2016construction}) for the case of the torus, that is, $g=1$.

\begin{figure}[ht]
\centering
\begin{tikzpicture}[scale=.5,%
->-/.style={decoration={
  markings,
  mark=at position .5 with {\arrow{>}}},postaction={decorate}},%
-<-/.style={decoration={
  markings,
  mark=at position .5 with {\arrow{<}}},postaction={decorate}},%
  ]
\draw (-15,0) node [anchor=east] {$\partial \Omega$};
\draw (2,0) node {$F$};
\draw [thick, brown] (-11.1,0) arc (180:0:.55 and .3);
\draw [thick, brown] (-10,0) arc (0:-90:.5 and 2.45);
\draw [thick, brown, dashed] (-11,-.25) arc (180:270:.5 and 2.4);
\draw [thick, green] (-8.9,0) arc (0:180:.55 and .3);
\draw [thick, green] (-10,0) arc (180:270:.5 and 2.45);
\draw [thick, green, dashed] (-9,-.25) arc (0:-90:.5 and 2.4);
\draw [thick, red] (-12,0) circle (2 and 1);
\draw [thick, blue] (-8,0) circle (2 and 1);
\draw [thick] (-10,0) circle (5 and 2.5)
              (-10.5,0.25) arc (0:-180:1.5 and .75)
              (-9.5,0.25) arc (-180:0:1.5 and .75)
              (-11,-0.25) arc (0:180:1 and .5)
              (-9,-0.25) arc (180:0:1 and .5);
\draw [red] (-3,1.5) node [anchor=east] {$a_1$} -- (-2,1.5);
\draw [blue] (-3,.5) node [anchor=east] {$a_2$} -- (-2,.5);
\draw [brown] (-3,-.5) node [anchor=east] {$b_1$} -- (-2,-.5);
\draw [green] (-3,-1.5) node [anchor=east] {$b_2$} -- (-2,-1.5);
\draw [thick,red,->-] ({2+2.5*cos(22.5)},{2.5*sin(22.5)}) -- ({2+2.5*cos(22.5+45)},{2.5*sin(22.5+45)}); 
\draw [thick,red,-<-] ({2+2.5*cos(22.5+90)},{2.5*sin(22.5+90)}) -- ({2+2.5*cos(22.5+135)},{2.5*sin(22.5+135)});
\draw [thick,brown,->-] ({2+2.5*cos(22.5+45)},{2.5*sin(22.5+45)}) -- ({2+2.5*cos(22.5+90)},{2.5*sin(22.5+90)}) ;
\draw [thick,brown,-<-] ({2+2.5*cos(22.5+135)},{2.5*sin(22.5+135)}) -- ({2+2.5*cos(22.5+180)},{2.5*sin(22.5+180)});
\draw [thick,blue,->-] ({2+2.5*cos(22.5+180)},{2.5*sin(22.5+180)}) -- ({2+2.5*cos(22.5+225)},{2.5*sin(22.5+225)});
\draw [thick,blue,-<-] ({2+2.5*cos(22.5+270)},{2.5*sin(22.5+270)}) -- ({2+2.5*cos(22.5+315)},{2.5*sin(22.5+315)});
\draw [thick,green,->-] ({2+2.5*cos(22.5+225)},{2.5*sin(22.5+225)}) -- ({2+2.5*cos(22.5+270)},{2.5*sin(22.5+270)});
\draw [thick,green,-<-]
({2+2.5*cos(22.5+315)},{2.5*sin(22.5+315)}) -- ({2+2.5*cos(22.5)},{2.5*sin(22.5)});
\end{tikzpicture}
\caption{The surface $\partial \Omega$ of genus $g=2$, its fundamental loops, and its fundamental polygon $F$.}
\label{fig:genus}
\end{figure}
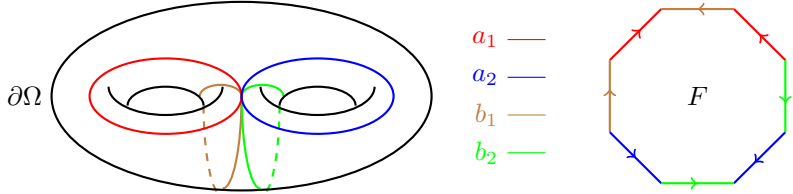

It remains to prove \eqref{item:frame}$\implies$\eqref{item:extension}. Without loss of generality, we may assume that $\Omega$ is the set obtained by thickening a bouquet $\cC$ of $g$ circles $(c_j)_{1 \le j \le g}$ meeting at a common point $\bk_0$. In this case, there exists a smooth function $R(\bK, t)$ defined for $\bK \in \partial \Omega$ and $t \in [0, 1]$, such that $R(\bK, t = 0) \in \cC$ while $R(\bK, t = 1) = \bK$, and such that $R$ is bijective from $ \partial \Omega \times (0, 1)$ to $\Omega \setminus \cC$ (think of smoothly shrinking the thickness of $\Omega$ to zero).

We first reduce to the case $Q(\bk) \equiv \mathrm{Id}$. As in \textbf{Step 2a} of the proof provided in Section~\ref{sec:chern}, we let $\Psi(\bk_0)$ be a frame of $Q(\bk_0)$. We first perform parallel transport along each of the circles $c_j$, and cure the holonomy acquired along the loop, \eg as in \cite{cornean2016construction}. This gives a smooth frame on $\cC = R(\partial \Omega, 0)$. We then parallel transport this frame along $R(\bK, t)$ for $t \in [0, 1]$ to get a smooth extension of the frame $\Psi$ to the whole $\Omega$. We conclude as before that we can work in the $\Psi$ basis, and assume that $Q \equiv \mathrm{Id}$.
  
We now extend the frame $\Phi$ for $P$ from $\partial \Omega$ to $\Omega$. By induction, and following the proof in Section~\ref{sec:chern}, it suffices to consider the case in which the rank of $P$ is $1$ and the dimension of the ambient Hilbert space is $M \ge 2$. In this case, the frame 
$\Phi$ consists of a smooth map $\phi \colon \partial \Omega \to \bbS^{2M-1}$, and Sard's lemma guarantees the existence of a point $-\phi^* \in \bbS^{2M-1}$ not in its image. Arguing as in~\eqref{eq:contractionTrick}, one can interpolate between $\phi(\bK)$ and $\phi^*$ following the rays of $R(\bK, t)$ from $t = 1$ to $t = 1/2$. We then further extend with the constant map $\phi^*$ in the rest of $\Omega$. The projection on the extended frame provides the desired extension of $P$.

\bigskip\bigskip

{\footnotesize

\begin{tabular}{rl}
(H.D. Cornean) & \textsc{Department of Mathematical Sciences, Aalborg University} \\
 &  Skjernvej 4A, 9220 Aalborg, Denmark \\
 &  \textsl{E-mail address}: \href{mailto:cornean@math.aau.dk}{\texttt{cornean@math.aau.dk}} \\
 \\
(A. Levitt) & \textsc{Inria Paris and Universit\'e Paris-Est, CERMICS (ENPC)} \\
 &  F-75589 Paris Cedex 12, France\\
 &  \textsl{E-mail address}: \href{mailto:antoine.levitt@inria.fr}{\texttt{antoine.levitt@inria.fr}} \\
  \\
(D. Monaco) & \textsc{Dipartimento di Matematica e Fisica, Universit\`{a} degli Studi di Roma Tre} \\
 &  Largo San Leonardo Murialdo 1, 00146 Rome, Italy \\
 &  \textsl{E-mail address}: \href{mailto:dmonaco@mat.uniroma3.it}{\texttt{dmonaco@mat.uniroma3.it}} \\
 \\
(D. Gontier) & \textsc{Universit\'e Paris-Dauphine, PSL Research University, CEREMADE} \\
& 75775 Paris, France \\
& \textsl{E-mail address}: \href{mailto:gontier@ceremade.dauphine.fr}{\texttt{gontier@ceremade.dauphine.fr}} \\
\end{tabular}

}

\begin{thebibliography}{00}
\bibitem{auckly2017parseval}
D.~Auckly D. and P.~Kuchment. On Parseval frames of exponentially decaying composite Wannier functions. Preprint available at \href{https://arxiv.org/abs/1704.05728}{\texttt{arXiv:1704.05728}} (2017).

\bibitem{brouder2007exponential}
Ch.~Brouder,  G.~Panati, M.~Calandra, Ch.~Mourougane, and N.~Marzari. Exponential localization of {W}annier functions in insulators. {\it Phys. Rev. Lett.} {\bf 98}(4), 046402 (2007).

\bibitem{cances2017robust}
\'{E}.~Canc\`{e}s, A.~Levitt, G.~Panati, and G.~Stoltz. Robust determination of maximally localized Wannier functions. {\it Phys. Rev. B} {\bf 95}(7), 075114 (2017).

\bibitem{cohen1966band}
M.L.~Cohen and T.K.~Bergstresser. Band structures and pseudopotential form factors for fourteen semiconductors of the diamond and zinc-blende structures. {\it Phys. Rev.} {\bf 141}(2), 789 (1966).

\bibitem{cornean2016construction}
H.D.~Cornean, I.~Herbst, and Gh.~Nenciu. On the construction of composite {W}annier functions. {\it Ann. Henri Poincar\`{e}} {\bf 17}(12), 3361--3398 (2016).

\bibitem{cornean2017_3D}
H.D.~Cornean and D.~Monaco. On the construction of Wannier functions in topological insulators: the 3D case. {\it Ann. Henri Poincar\`{e}} {\bf 18}(12), 3863--3902 (2017).

\bibitem{cornean2017parseval}
H.D.~Cornean and D.~Monaco. Parseval frames of localized Wannier functions. Preprint available at \href{https://arxiv.org/abs/1704.00932}{\texttt{arXiv:1704.00932}} (2017).

\bibitem{cornean2017wannier}
H.D.~Cornean, D.~Monaco, and S.~Teufel. Wannier functions and $\mathbb{Z}_2$ invariants in time-reversal symmetric topological insulators. {\it Rev. Math. Phys.} {\bf 29}(2), 1730001 (2017).

\bibitem{damle2017}
A.~Damle, A.~Levitt, and L.~Lin. Variational formulation for {W}annier functions With entangled band structure. In preparation.

\bibitem{damle2017scdm}
A.~Damle, L.~Lin, and L.~Ying. SCDM-k: Localized orbitals for solids via selected columns of the density matrix. {\it J. Comput. Phys.} {\bf 334}, 1--15 (2017).

\bibitem{fiorenza2016construction}
D.~Fiorenza, D.~Monaco, and G.~Panati. Construction of real-valued localized composite {W}annier functions for insulators. {\it Ann.~Henri Poincar{\'e}} {\bf 17}(1), 63--97 (2016).

\bibitem{friedan1982proof}
D.~Friedan. A proof of the {N}ielsen-{N}inomiya theorem. {\it Commun. Math. Phys.} {\bf 85}(4), 481--490 (1982).

\bibitem{fefferman2012honeycomb}
C.~Fefferman  and M.~Weinstein. Honeycomb lattice potentials and {D}irac points. {\it J. Amer. Math. Soc.} {\bf 25}(4), 1169--1220 (2012).

\bibitem{Gontier2017todo}
D.~Gontier, A.~Levitt, and S.~Siraj-Dine. Construction of Wannier functions by matrix homotopy. In preparation.

\bibitem{guillemin1974topology}
V.~Guillemin and A.~Pollack. {\it Differential Topology}. Prentice Hall, 1974.

\bibitem{jost2013compact}
J.~Jost. {\it Compact {R}iemann Surfaces: An Introduction to Contemporary Mathematics}. Springer, 2013.

\bibitem{marzari2012maximally}
N.~Marzari, A.A.~Mostofi, J.R.~Yates, I.~Souza, and D.~Vanderbilt. Maximally localized Wannier functions: Theory and applications. {\it Rev. Mod. Phys.} {\bf 84}(4), {1419} (2012).

\bibitem{marzari1997maximally}
N.~Marzari and D.~Vanderbilt. Maximally localized generalized {W}annier functions for composite energy bands. {\it Phys. Rev. B} {\bf 56}(20), 12847 (1997).

\bibitem{monaco2017chern}
D.~Monaco. Chern and Fu--Kane--Mele invariants as topological obstructions. Chap.~12 in: G.~Dell'Antonio and A.~Michelangeli (eds.), {\it Advances in Quantum Mechanics: Contemporary Trends and Open Problems}. Vol.~18 in Springer INdAM Series. Springer, 2017.

\bibitem{monaco2014graphene}
D.~Monaco and G.~Panati. Topological invariants of eigenvalue intersections and decrease of Wannier functions in graphene. {\it J. Stat. Phys.} {\bf 155}(6), 1027–-1071 (2014). 

\bibitem{monaco2015symmetry}
D.~Monaco and G.~Panati. Symmetry and localization in periodic crystals: triviality of {B}loch bundles with a fermionic time-reversal symmetry. {\it Acta Appl. Math.} {\bf 137}(1), 185--203 (2015).

\bibitem{mostofi2008wannier90}
A.A.~Mostofi, J.R.~Yates, Y.-S.~Lee, I.~Souza, D.~Vanderbilt, and N.~Marzari. wannier90: A tool for obtaining maximally-localised Wannier functions. {\it Comput. Phys. Commun.} {\bf 178}(9), 685--699 (2008).

\bibitem{mathai2017global}
V.~Mathai and G. Ch.~Thiang. Global topology of {W}eyl semimetals and {F}ermi arcs. {\it J. Phys. A} {\bf 50}(11), 11LT01 (2017).

\bibitem{mustafa2015automated}
J.I.~Mustafa, S.~Coh, M.L.~Cohen, and S.G.~Louie. Automated construction of maximally localized Wannier functions: Optimized projection functions method. {\it Phys. Rev. B} {\bf 92}(16), 165134 (2015).

\bibitem{nielsen1981no}
H.B.~Nielsen and M.~Ninomiya. A no-go theorem for regularizing chiral fermions. {\it Phys. Lett. B} {\bf 105}(2-3), 219--223 (1981).

\bibitem{panati2007triviality}
G.~Panati, Triviality of {B}loch and {B}loch-{D}irac bundles. {\it Ann. Henri Poincar\'{e}} {\bf 8}(5), 995--1011 (2007).

\bibitem{reed1978analysis}
M.~Reed and B.~Simon. {\it Methods of Modern Mathematical Physics Vol. IV: Analysis of Operators}. Academic Press, 1978.

\bibitem{simon1983holonomy}
B.~Simon. Holonomy, the quantum adiabatic theorem, and {B}erry's phase. {\it Phys. Rev. Lett.} {\bf 51}(24), 2167 (1983).

\bibitem{souza2001maximally}
I.~Souza, N.~Marzari, and D.~Vanderbilt. Maximally localized {W}annier functions for entangled energy bands. {\it Phys. Rev. B} {\bf 65}(3), 035109 (2001).

\bibitem{neumann2000behaviour}
J. von Neumann and E.~Wigner. On the behaviour of eigenvalues in adiabatic processes. {\it Phys. Z.} {\bf 30}, 467 (1929). Republished in: Hettema, H. (ed.), {\it Quantum Chemistry: Classic Scientific Papers}, pp.~25--31. World Scientific, 2000.

\bibitem{yates2007spectral}
J.R.~Yates, X.~Wang, D.~Vanderbilt, and I.~Souza. Spectral and {F}ermi surface properties from {W}annier interpolation. {\it Phys. Rev. B} {\bf 75}(19), 195121 (2007).

\end{thebibliography}
\end{document}